\documentclass[preprint,1p,times]{elsarticle}

\input{diagrams}
\usepackage{graphicx}
  \graphicspath{{./pics/}}
  \DeclareGraphicsExtensions{.pdf}
\usepackage{amssymb}
\usepackage{amsmath}
\usepackage[caption=false,font=footnotesize]{subfig}

\newcommand{\eps}{\varepsilon}
\newtheorem{theorem}{Theorem}

\newtheorem{definition}[theorem]{Definition}

\newproof{proof}{Proof}

\journal{a journal.}

\begin{document}

\begin{frontmatter}

\title{A Note on Undecidability of Observation Consistency\\ for Non-Regular Languages}

\author{Tom{\' a}{\v s}~Masopust}
  \ead{masopust@math.cas.cz}
  \address{Institute of Mathematics, Academy of Sciences of the Czech Republic\\ {\v Z}i{\v z}kova 22, 616 62 Brno, Czech Republic}
\cortext[cor1]{Corresponding author at: Institute of Mathematics, Academy of Sciences of the Czech Republic, {\v Z}i{\v z}kova 22, 616 62 Brno, Czech Republic, Tel.~+420532290376, Fax.~+420541218657 }

\begin{abstract}
  One of the most interesting questions concerning hierarchical control of discrete-event systems with partial observations is a condition under which the language observability is preserved between the original and the abstracted plant. Recently, we have characterized two such sufficient conditions---observation consistency and local observation consistency. In this paper, we prove that the condition of observation consistency is undecidable for non-regular (linear, deterministic context-free) languages. The question whether the condition is decidable for regular languages is open.
\end{abstract}
\end{frontmatter}

\section{Introduction}
  The main issue in supervisory control of discrete-event systems~\cite{CL08} is the state-space explosion problem inherent to large systems, which makes the standard techniques that compute and use the whole system model very difficult and often impossible to use. Methods how to decrease the complexity are intensively studied in the literature. Modular control and hierarchical control are the most successful approaches known so far. These two approaches are complementary because the modular approach can be understood as a horizontal modularity, while the hierarchical approach can be understood as a vertical modularity. The best known results are achieved when the two approaches are combined~\cite{KS}. 
  
  During the last few decades, hierarchical control of discrete-event systems with complete observations has widely been investigated. Several important concepts---the {\em observer property\/} \cite{WW96}, {\em output control consistency\/} (OCC) \cite{WZ91}, and {\em local control consistency\/} (LCC) \cite{SB08}---have been proposed and studied. These concepts are sufficient conditions for the high-level synthesis of a nonblocking and optimal supervisor having a low-level implementation. Recently, we have addressed hierarchical control of partially observed discrete-event systems. In \cite{KM10}, we have presented a sufficient condition which ensures that the optimal high-level supervisor with partial observation is implementable in the original low-level plant. However, the condition imposes that all the observable events must be included in the high-level alphabet, which is very restrictive. Later, in \cite{cdc-ecc2011}, we have discussed a weaker, less restrictive condition, and we have introduced two new structural conditions for projections called {\em local observation consistency\/} (LOC) and {\em observation consistency\/} (OC). The latter addresses a certain consistency between observations on the high level and the low level, and the former is an extension of the observer property under partial observations. We have shown that projections which satisfy OC, LOC, LCC, and which are observers are also suitable for the nonblocking least restrictive hierarchical control under partial observation. However, we have left the question whether the conditions are decidable or not open. 
  
  In this paper, we prove that the condition of observation consistency is undecidable for non-regular (linear, deterministic context-free) languages. The motivation to study this case comes from the fact that although supervisory control of discrete-event systems is mostly developed for regular languages, several attempts of its generalization to deterministic context-free languages have appeared in the literature~\cite{gri07,gri10}. However, the fundamental problem whether the condition is decidable for regular languages is still unsolved.

\section{Preliminaries and definitions}
	In this paper, we assume that the reader is familiar with the basic concepts of supervisory control theory \cite{CL08} and automata and formal language theory \cite{salomaa}. For an alphabet $\Sigma$, defined as a finite nonempty set, $\Sigma^*$ denotes the free monoid generated by $\Sigma$, where the unit of $\Sigma^*$, the empty string, is denoted by $\eps$. A {\em language\/} over $\Sigma$ is a subset of $\Sigma^*$. A {\em (natural) projection\/} $P: \Sigma^* \to \Sigma_0^*$, where $\Sigma$ and $\Sigma_0\subseteq \Sigma$ are alphabets, is a homomorphism defined so that $P(a) = \eps$ for $a\in \Sigma\setminus \Sigma_0$, and $P(a) = a$ for $a\in \Sigma_0$. The {\em inverse image} of $P$, denoted by $P^{-1} : \Sigma_0^* \to 2^{\Sigma^*}$, is defined as $P^{-1}(a)=\{s\in \Sigma^* \mid P(s) = a\}$. These definitions can naturally be extended to languages. A string $s\in \Sigma^*$ is a {\em prefix} of a string $w\in \Sigma^*$ if $w=st$, for some $t\in \Sigma^*$. The prefix closure $\overline{L}=\{w\in \Sigma^* \mid \text{there exists } v\in \Sigma^* \text{ such that } wv\in L\}$ of a language $L\subseteq \Sigma^*$ is the set of all prefixes of all its elements. A language $L$ is prefix-closed if $L=\overline{L}$.

  In this paper, the notion of a generator is used to denote an incomplete deterministic finite automaton. A {\em generator\/} $G$ is a quintuple $G=(Q,\Sigma,\delta,q_0,F)$, where $Q$ is a finite set of {\em states}, $\Sigma$ is an {\em input alphabet}, $\delta: Q \times \Sigma \to Q$ is a {\em partial transition function}, $q_0 \in Q$ is the {\em initial state}, and $F\subseteq Q$ is the set of {\em final or marked states}. In the usual way, $\delta$ is extended to a function from $Q \times \Sigma^*$ to $Q$. The language {\em generated\/} by the generator $G$ is defined as the set of all possible strings $G$ can read from the initial state, that is, $L(G) = \{w\in \Sigma^* \mid \delta(q_0,w)\in Q\}$, and the language {\em marked\/} by the generator $G$ is defined as the set of all strings leading $G$ from the initial state to a marked state, that is, $L_m(G) = \{w\in \Sigma^* \mid \delta(q_0,w)\in F\}$. Note that, by definition, $L_m(G)\subseteq L(G)$, and $L(G)$ is always prefix-closed. Moreover, we use the predicate $\delta(q,a)!$ to denote that the transition $\delta(q,a)$ is defined in $G$. 

	Let $L_1\subseteq E_1^*$ and $L_2\subseteq E_2^*$ be two languages. The {\em parallel composition of $L_1$ and $L_2$\/} is defined as the language $L_1\parallel L_2 = P_1^{-1}(L_1) \cap P_2^{-1}(L_2)$. For the corresponding automata definition, the reader is referred to~\cite{CL08}. 

  Let $G$ be a generator over an alphabet $\Sigma$, and let $\Sigma_u\subseteq \Sigma$ be the subset of all uncontrollable events. A language $K\subseteq\Sigma^*$ is {\em controllable} with respect to $L(G)$ and $\Sigma_u$ if $\overline{K}\Sigma_u\cap L(G)\subseteq \overline{K}$. Moreover, $K$ is $L_m(G)$-closed if $K = \overline K \cap L_m(G)$. Furthermore, let $\Sigma_c=\Sigma\setminus\Sigma_u$ be the subset of all controllable events, and let $\Sigma_o\subseteq\Sigma$ be the set of all observable events with $P$ as the corresponding projection from $\Sigma^*$ to $\Sigma_o^*$. The language $K\subseteq L(G)$ is {\em observable\/} with respect to $L(G)$, $\Sigma_o$, and $\Sigma_c$ if for all $s,s'\in L(G)$ such that $P(s)=P(s')$ and for all $e\in \Sigma_c$, $(se \in L(G) \wedge s'e \in \overline{K} \wedge s \in \overline{K}) \Rightarrow se \in \overline{K}$. Algorithms for these properties can be found in~\cite{CL08}.

  Given a system $G$ over an alphabet $\Sigma$ and a specification language $K\subseteq L_m(G)$, the aim of supervisory control is to find a nonblocking supervisor $S$ such that the closed-loop system $S/G$ satisfies the specification and is nonblocking, that is, $\overline{L_m(S/G)} = L(S/G)=K$; as these notions are not important for the understanding of this paper, we do not discuss them here and refer the reader to~\cite{CL08,Won04} for more details. We only note that it is known that such a supervisor exists if and only if $K$ is controllable with respect to $L(G)$ and $\Sigma_u$, $L_m(G)$-closed, and observable with respect to $L(G)$, $\Sigma_o$, and $\Sigma_c$.

\section{Observation consistency}
  Recently, we have studied the problem of an existence of supervisors under partial observation based on the computation of abstractions. In this framework, the plant is represented as a generator $G$ over an alphabet $\Sigma$ and it is desired to realize a high-level specification $K \subseteq \Sigma_{hi}^*$, where $\Sigma_{hi}\subset\Sigma$ is a high-level alphabet. Our recent result is recalled below as Theorem~\ref{thm3}. 
  
  For projections and abstractions, we use the following notations:
  $P:\Sigma^* \to \Sigma^*_o$,
  $A:\Sigma^* \to (\Sigma_{hi})^*$,
  $P_{hi}:(\Sigma_{hi})^{*} \to (\Sigma_{hi}\cap \Sigma_o)^*$, and
  $A_o:\Sigma_o^* \to (\Sigma_{hi}\cap \Sigma_o)^*$ as illustrated in the commutative diagram in Figure~\ref{projections}.
  \begin{figure}[ht]
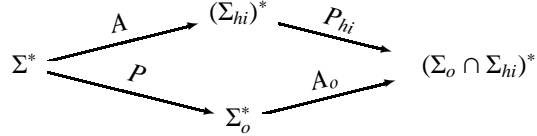

    \begin{diagram}[height=1em,width=4em]
                &           & (\Sigma_{hi})^* &                & \\
                & \ruTo^{A} &                 & \rdTo^{P_{hi}} & \\
      \Sigma^*  &           &                 &                & \relax \qquad (\Sigma_o\cap \Sigma_{hi})^* \\
                & \rdTo^{P} &                 & \ruTo^{A_{o}}  & \\
                &           & \Sigma_{o}^*    &                & \\
    \end{diagram}
    \caption{Commutative diagram of abstractions and projections.}
    \label{projections}
  \end{figure}

  \begin{definition}[Observation consistency]\label{def:consistency}
    A language $L = \overline L \subseteq \Sigma^*$ is said to be {\em observation consistent} with respect to projections $A$, $P$, and $P_{hi}$ if for all strings $t,t' \in A(L)$ such that $P_{hi}(t) = P_{hi}(t')$, there exist strings $s,s' \in L$ such that $A(s) = t$, $A(s') = t'$, and $P(s) = P(s')$.
  \end{definition}  
  
  Thus, observation consistency requires that any two strings that have the same observation in the abstracted high-level plant have also corresponding strings with the same observation in the original low-level plant.

  The other condition required for the next theorem is the local observation consistency.
  \begin{definition}[Local observation consistency]\label{def:loc}
    A language $L  = \overline L \subseteq \Sigma^*$ is said to be {\em locally observation consistent} with respect to projections $A$, $P$, and the set of controllable events $\Sigma_c$ if for all strings  $s,s'\in L$ and events $e\in \Sigma_c \cap \Sigma_{hi}$ such that $A(s)e \in A(L)$, $A(s')e \in A(L)$, and $P(s) = P(s')$, there exist $u,u'\in (\Sigma\setminus \Sigma_{hi})^*$ such that $P(u) = P(u')$ and $sue \in L$ and $s'u'e\in L$.
  \end{definition}
  
  Finally, recall that two languages $L_1\subseteq \Sigma_1*$ and $L_2 \subseteq \Sigma_2^*$ are {\em synchronously nonconflicting\/} if $\overline{L_{1}\parallel L_{2}} = \overline{L_{1}} \parallel \overline{L_{2}}$.
  \begin{theorem}[\cite{cdc-ecc2011}]\label{thm3}
    Let $G$ be a generator over an alphabet $\Sigma$, and let $K\subseteq A(L(G))$ be a high-level specification. Assume that $L(G)$ is observation consistent with respect to projections $A$, $P$, and $P_{hi}$, that $K$ and $L(G)$ are synchronously nonconflicting, and that $L(G)$ is locally observation consistent with respect to $A$, $P$, and $\Sigma_{c}$. Then, the language $K$ is observable with respect to $A(L(G))$, $\Sigma_{hi} \cap \Sigma_o$, and $\Sigma_{hi} \cap \Sigma_c$ if and only if the language $K\parallel L(G)$ is observable with respect to $L(G)$, $\Sigma_o$, and $\Sigma_c$.
  \end{theorem}

\section{Main Result}
  In this section, we prove that if the plant language $L(G)$ is a non-regular language, even though it is only a linear, deterministic context-free language, the observation consistency condition is undecidable. For the definitions of linear and deterministic context-free languages, the reader is referred to~\cite{salomaa}.
  
  \begin{theorem}
    The observation consistency condition for linear, deterministic context-free languages is undecidable.
  \end{theorem}
  \begin{proof}
    We prove the theorem by reduction of Post's Correspondence Problem (PCP) to the problem of observation consistency. Recall that PCP is the problem whether, given two finite sets $A=\{w_1,w_2,\ldots,w_n\}$ and $B=\{u_1,u_2,\ldots,u_n\}$ of $n$ strings over an alphabet $\Sigma$, there exists a sequence of indices $i_1 i_2 \ldots i_k$, for $k\ge 1$, such that $w_{i_1} w_{i_2} \ldots w_{i_k} = u_{i_1} u_{i_2} \ldots u_{i_k}$. It is well-known that PCP is undecidable~\cite{post}.
    
    Let $\{w_1,w_2,\ldots,w_n\}$ and $\{u_1,u_2,\ldots,u_n\}$ be an instance of PCP over an alphabet $\Sigma$ such that for all $i=1,2,\ldots, n$, we have $w_i\neq u_i$. Let $E=\{1,2,\ldots,n\}$ be a new alphabet, that is, $E\cap\Sigma=\emptyset$. We use the notation $w^R$ to denote the reversal or mirror image of a string $w\in\Sigma^*$. Define the language 
    \[
      L=\{@i_1i_2\ldots i_m\$w_{i_m}^R\ldots w_{i_2}^Rw_{i_1}^R@ \mid m\ge 1\}\cup
         \{i_1i_2\ldots i_m\$u_{i_m}^R\ldots u_{i_2}^Ru_{i_1}^R\# \mid m\ge 1\}\,.
    \]
    Note that this language is linear and deterministic context-free, and it is also not hard to see that the language $\overline{L}$ is linear and deterministic context-free, too. The linearity is obvious from the form of the words, and a deterministic pushdown automaton works so that based on $@$ it distinguishes the two parts of the language, and then it pushes the indices to the pushdown and after reading $\$$ it pops indices from the pushdown which tells the automaton what strings should be read from the input. 
    
    Finally, we define the abstraction $A:(\Sigma\cup\{@,\#,\$\}\cup E)^*\to \{@,\#\}^*$ and the projection $P:(\Sigma\cup\{@,\#,\$\}\cup E)^*\to (\Sigma\cup E)^*$. Now, we prove that PCP has a solution if and only if the language $\overline{L}$ satisfies the observation consistency condition. Note that from the definition of the abstraction and projection, it follows that for any two strings $t,t'\in A(\overline{L})=\{@,@@,\#,\eps\}$, it holds that $P_{hi}(t)=\eps=P_{hi}(t')$.
    
    Assume that PCP has a solution, say $i_1i_2\ldots i_k$ with $w_{i_1}w_{i_2}\ldots w_{i_k}=u_{i_1}u_{i_2}\ldots u_{i_k}$. Then, if $t=t'$, there exists $s=s'$ such that $A(s)=A(s')=t=t'$ and, obviously, $P(s)=P(s')$. Thus, assume that $t\neq t'$. We have six possibilities for $t$ and $t'$, namely
    \begin{enumerate}
      \item $t=@$ and $t'=@@$: In this case, set $s=@1\$w_1^R$ and $s'=@1\$w_1^R@$. Then, $A(s)=@$, $A(s')=@@$, and $P(s)=1w_1^R=P(s')$ as required.
      
      \item $t=@$ and $t'=\#$: Set $s=@i_1i_2\ldots i_k\$w_{i_k}^R\ldots w_{i_2}^Rw_{i_1}^R$ and $s'=i_1i_2\ldots i_k\$u_{i_k}^R\ldots u_{i_2}^Ru_{i_1}^R\#$. Then, $A(s)=@$, $A(s')=\#$, and $P(s)=i_1i_2\ldots i_kw_{i_k}^R\ldots w_{i_2}^Rw_{i_1}^R=i_1i_2\ldots i_ku_{i_k}^R\ldots u_{i_2}^Ru_{i_1}^R=P(s')$.
      
      \item $t=@$ and $t'=\eps$: Set $s=@i_1i_2\ldots i_k\$w_{i_k}^R\ldots w_{i_2}^Rw_{i_1}^R$ and $s'=i_1i_2\ldots i_k\$u_{i_k}^R\ldots u_{i_2}^Ru_{i_1}^R$. Then, $A(s)=@$, $A(s')=\eps$, and $P(s)=i_1i_2\ldots i_kw_{i_k}^R\ldots w_{i_2}^Rw_{i_1}^R=i_1i_2\ldots i_ku_{i_k}^R\ldots u_{i_2}^Ru_{i_1}^R=P(s')$.
      
      \item $t=@@$ and $t'=\#$: Set $s=@i_1i_2\ldots i_k\$w_{i_k}^R\ldots w_{i_2}^Rw_{i_1}^R@$ and $s'=i_1i_2\ldots i_k\$u_{i_k}^R\ldots u_{i_2}^Ru_{i_1}^R\#$. Then, $A(s)=@@$, $A(s')=\#$, and $P(s)=i_1i_2\ldots i_kw_{i_k}^R\ldots w_{i_2}^Rw_{i_1}^R=i_1i_2\ldots i_ku_{i_k}^R\ldots u_{i_2}^Ru_{i_1}^R=P(s')$.
      
      \item $t=@@$ and $t'=\eps$: Set $s=@i_1i_2\ldots i_k\$w_{i_k}^R\ldots w_{i_2}^Rw_{i_1}^R@$ and $s'=i_1i_2\ldots i_k\$u_{i_k}^R\ldots u_{i_2}^Ru_{i_1}^R$. Then, $A(s)=@@$, $A(s')=\eps$, and $P(s)=i_1i_2\ldots i_kw_{i_k}^R\ldots w_{i_2}^Rw_{i_1}^R=i_1i_2\ldots i_ku_{i_k}^R\ldots u_{i_2}^Ru_{i_1}^R=P(s')$.
      
      \item $t=\#$ and $t'=\eps$: Set $s=1\$w_1^R\#$ and $s'=1\$w_1^R$. Then, $A(s)=\#$, $A(s')=\eps$, and $P(s)=1w_1^R=P(s')$.
    \end{enumerate}
    Thus, we have shown that if PCP has a solution, the language $\overline{L}$ satisfies the observation consistency condition.
    
    On the other hand, assume that the instance of PCP has no solution. Then, we prove that for $t=@@$ and $t'=\#$, there are no $s$ and $s'$ in $\overline{L}$ such that $A(s)=@@$, $A(s')=\#$, and $P(s)=P(s')$, that is, that the language $\overline{L}$ does not satisfy the observation consistency condition. For the sake of contradiction, assume that there exist such $s$ and $s'$ in $\overline{L}$. Let $s$ be of a form $@i_1i_2\ldots i_k\$w_{i_k}^R\ldots w_{i_2}^Rw_{i_1}^R@$ and $s'$ be of a form $j_1j_2\ldots j_{k'}\$u_{j_{k'}}^R\ldots u_{j_2}^Ru_{j_1}^R\#$, which are the only forms of strings with abstractions $@@$ and $\#$, respectively. Then, by our assumption, $A(s)=@@$, $A(s')=\#$, and $P(s)=i_1i_2\ldots i_kw_{i_k}^R\ldots w_{i_2}^Rw_{i_k}^R=j_1j_2\ldots j_{k'}u_{i_{k'}}^R\ldots u_{i_2}^Ru_{i_1}^R=P(s')$. However, this means that $i_1i_2\ldots i_k=j_1j_2\ldots j_{k'}$, which implies that $k=k'$ and $i_z=j_z$ for $1\le z\le k$, and $w_{i_k}^R\ldots w_{i_2}^Rw_{i_1}^R=u_{i_k}^R\ldots u_{i_2}^Ru_{i_1}^R$, which means that $w_{i_1}w_{i_2}\ldots w_{i_k}=u_{i_1}u_{i_2}\ldots u_{i_k}$. But this is a solution of our instance of PCP, namely the sequence $i_1 i_2 \ldots i_k$, and it is a contradiction. Thus, there are no such strings $s$ and $s'$ for $t=@@$ and $t'=\#$. Hence, the instance of PCP has a solution if and only if the language $\overline{L}$ satisfies the observation consistency condition, which means that observation consistency is undecidable for linear, deterministic context-free languages.
  \qed\end{proof}
  
\section{Conclusion}\label{conclusion}
  In this paper, we have shown that if the language is linear, deterministic context-free, then the observation consistency condition is undecidable. However, it needs to be mentioned that no algorithm is known to decide the observation consistency condition even for regular languages. More specifically, it is an open problem whether the condition of observation consistency is decidable for regular languages. This condition is of great interest in hierarchical control with partial observation, and so is the decidability problem. Moreover, if it is proven undecidable, a stronger condition that implies observation consistency, is decidable, and does not imply that all observable events must be included in the high-level alphabet is of great interest. 

\section*{Acknowledgments}
  The research received funding from the European Community's Seventh Framework Programme under grant agreement no. INFSO-ICT-224498, from the Academy of Sciences of the Czech Republic, Institutional Research Plan no. AV0Z10190503, and from the GA{\v C}R grant no. P202/11/P028.

\bibliographystyle{model1-num-names}
\bibliography{biblio}

\end{document}